\documentclass[english,aps,pra,showpacs,superscriptaddress,floatfix, notitlepage,reprint,unicode=true,colorlinks=true,citecolor=Blue,linkcolor=RubineRed,urlcolor=Blue
]{revtex4-2}
\pdfoutput=1
\usepackage[T1]{fontenc}
\usepackage[utf8]{inputenc}
\usepackage{color}
\usepackage{babel}
\usepackage{mathrsfs}
\usepackage{amsmath}
\usepackage{amsthm}
\usepackage{amssymb}
\usepackage{natbib }
\usepackage{bm}
\usepackage[unicode=true]
{hyperref}

\makeatletter
\usepackage{braket}
\usepackage{txfonts} 
\usepackage{graphicx}

\usepackage[usenames,dvipsnames]{xcolor}

\date{\today}

\makeatother
\newcommand{\Qu}[1]{[#1]}

\newtheorem{thm}{Theorem}
\newtheorem{prp}{Proposition}

\begin{document}
	
	\preprint{APS/123-QED}
	
	\title{Quotient Space Quantum Codes}
	
	\author{Jing-Lei Xia~\href{https://orcid.org/0000-0001-5539-0198}	}%
	\email[Corresponding author.\\]{JingLei\_Xia@163.com}
	
	\affiliation{
		China Coast Guard Academy,~Ningbo,~Zhejiang,~315801,~China
	}%

	\date{\today}

	\begin{abstract}
		Additive codes and some nonadditive codes use the single and multiple invariant subspaces of the stabilizer G, respectively, to construct quantum codes, so the selection of the invariant subspaces is a key problem. In this paper, I provide the necessary and sufficient conditions for this problem and,  establish the quotient space codes to construct quantum codes. These new codes unify additive codes and codeword stabilized codes and can transmit classical codewords. Actually, I give an alternative approach to constructing union stabilizer codes, which is different from that of Markus Grassl and Martin Roetteler, and which is easier to deal with degenerate codes. I also present new bounds for quantum codes and provide a simple proof of the quantum Singleton bound. The quotient space approach provides a concise and clear mathematical framework for the study of quantum error-correcting codes.
	\end{abstract}
	\maketitle

	\textit{Introduction}.\textit{---}	Quantum error-correcting codes (QECCs) \cite{terhal2015quantum} are a coding scheme used to protect quantum information from noise and errors and are a crucial problem for the realization of quantum computing and quantum communication \cite{BCHS98,nielsen2010quantum}. Since Shor \cite{shor1995scheme} constructed the first code, a large number of coding schemes have been proposed, among which stabilizer codes or additive codes \cite{gottesman1997stabilizer,calderbank1997quantum,CRSS98} are the most representative class of codes. They use a $2^k$-dimensional invariant subspace of an additive subgroup $G$ (stabilizer) of the quantum error operator group $E$ to construct codes. Additive codes are usually denoted as $[[n,k,d]],$ indicating that $n$ physical qubits are encoded as $k$ logical qubits and can correct $\lfloor\frac{d-1}{2}\rfloor$-qubit errors, where $d$ is the minimum distance. 
	
	Another type of code is nonadditive codes, first proposed by Rains et al.\cite{rains1997nonadditive}, some of which use $L$ one-dimensional invariant subspaces of the stabilizer to construct codes, denoted as $((n,L,d)),$ representing $L$-dimensional subspaces in an $n$-dimensional Hilbert space, that can correct $\lfloor\frac{d-1}{2}\rfloor$ -qubit errors.
	The references \cite{smolin2007simple,yu2008nonadditive,AVCA08}  used different methods to construct nonadditive codes that are better than additive codes. Reference \cite{cross2008codeword} created codeword stabilized codes (CWS) that unify additive and some good nonadditive codes. These methods are elegant but seem too esoteric and lack a language or concepts to succinctly express the problem, such as how to construct  codes for $d\geq 4.$
	
	This paper utilizes normed quotient spaces to establish quotient space codes. If a quotient space code $(n-k, L, d)$ that satisfies the measurement conditions \eqref{eq:measuring Ed} exists, then there are $L$ additive codes of dimension $2^k$ that together form a quotient space quantum code $((n, 2^kL, d)).$ When $L=1,$ an additive code is obtained, and when $k=0,$ a CWS code is obtained. The circuit design and decoding scheme of this code can be improved based on the additive code. It is observed that although CWS codes can include additive codes, these two code families have fundamentally different error correction methods.~I also provide new bounds-\textit{measurement bounds},~which reflect the unique properties of quantum codes compared to classical codes and are applicable to both degenerate and nondegenerate codes, and give a novel  perspective on the open question of whether any code does not violate the quantum Hamming bound\cite{DAL22}.
	The choice of coset representatives and norms greatly simplifies the problem description, as shown in the proof of the quantum Singleton bound. I present four methods for constructing the code $((8, 8, 3))$ and provide an example of a degenerate CWS code $((12, 2, 5)).$
	
	Coincidentally, when I finished the main work of this paper, I realized that 
	Markus Grassl and Thomas Beth had already considered the construction of codes with multiple subspaces in 1997~\cite{grassl1997notenonadditivequantumcodes}, and Markus Grassl and Martin Roetteler subsequently developed and refined the theory of \textit{ union stabilizer} (USt) codes in 2008~\cite{GR08} and 2013~\cite{GR2013}.~But our approaches are fundamentally different.~The USt construction first selects the cosets and then calculates the distance d, and it is known in the framework of QSQC that such an order can easily lead to a small or even zero distance in the case of degenerate codes, hence they were not aware of any degenerate truly nonadditive code whose minimum distance is
	strictly larger than $d_{min}(\Omega_{*})$ where  $ \Omega_{*} $ is  classical union normalizer code\cite{GR2013}. At the end of the paper, I give the detailed relationship between these two methods and give two degenerate codes~$ ((9,64,2))$ and $((7,16,2))$ that satisfy a strict inequality in Eq. (10.6) in \cite{GR2013}.

	\textit{Preparation.}\textit{---}First, I introduce the basics and some notations of quantum codes. The quantum error group $E_n$ is written as $E_n=\{i^qX(a)Z(b)\mid 0\leq q\leq 3,~a,~b\in \mathbb{F}_2^n\},$ and in situations where there is no confusion it is also written as $E.$ We define the mapping
	$\label{equ:hom}
	\varphi:E_n\rightarrow \mathbb{F}_2^{2n},~e=i^q X(a)Z(b)\longmapsto\varphi(e)=(a|b),
	$ which induces a group isomorphism  $
	\bar{\varphi}:\overline{E}_n=E_n/\{\pm 1,\pm i\}\rightarrow V=\mathbb{F}_2^{2n},~\bar{\varphi}(\bar{e})=(a|b),
	$
	thus allowing quantum errors to be represented by vectors in $ \mathbb{F}_2^{2n},$ and we may denote $\bar{e}=\varphi(e).$ Let $ w_Q(\bar{e}) $ be the quantum weight of $ e ,$ $ w_H(\bar{e}) $ be the Hamming weight of $ e ,$ $ E(d) $ be the quantum errors with weight not exceeding $ d ,$ and $\overline{E}(d)$ be its image. The symplectic inner product of vectors in $ \mathbb{F}_2^{2n} $ is defined as $(c_1,c_2)_s=a·b'+a'·b,~c_1=(a|b),~c_2=(a'|b')\in  \mathbb{F}^{2n}_2.~C$ is a subspace of $ \mathbb{F}^{2n}_2 ,$ and $  C^{\bot}_s=\{c_s\mid(c_s,c)_s=0,~\forall ~c\in C\} ,$ is the symplectic orthogonal space of $ C .$ $G$ is an additive subgroup of $ E $ and $\overline{G}=C,~G^{\bot}_s=\{e\mid e\in E_n,~\bar{e}\in C^{\bot}_s\},d_m=min\{w_Q(c)\mid c\in(C^{\bot}_s)\backslash C\},~C(d)=\overline{E}(d)\bigcap C,~ C(d)^{\bot}_s $ is the symplectic orthogonal space of $ C(d) ,$ and it is known that $C^{\bot}_s\subset C(d)^{\bot}_s.$ A characteristic of a finite commutative group $ G $ is a group homomorphism from $ G $ to the nonzero complex multiplication group $ \mathbb{C}^{*} $
	\begin{equation*}
		\lambda:G\rightarrow \mathbb{C}^{*},~\lambda(gg')=\lambda(g)\lambda(g')~(g,~g'\in G) .	
	\end{equation*} 
	The set of all characteristics of $ G $ forms a group $\widehat{G},$ and the characteristic group of $ C $ is denoted as $\widehat{C}.$
	
	\textit{Characterization of Invariant Subspaces.} \textit{---} To express the problem in concise mathematical form, it is necessary to characterize the invariant subspaces of $G$ and the action of quantum errors on them. Due to the isomorphism between a  finite abelian group $G$ and its character group $\widehat{G}$ \cite{jacobson2009basic}, and the natural isomorphism between the additive character group $\widehat{C}$ of $C$ and the additive quotient group $\mathbb{F}_2^{2n}/C^{\bot}_s$ \cite{feng2010quantum}, it is easy to prove using the homomorphism theorem that
	\begin{equation}\label{key}
		\widehat{G}\cong G\cong C \cong \widehat{C}\cong\dfrac{\mathbb{F}^{2n}_2}{C^{\bot}_s}\cong \dfrac{E_n}{G^{\bot}_s} . 
	\end{equation}
	Let $I=\frac{\mathbb{F}^{2n}_2}{C^{\bot}_s},$ and $\Qu{i}$ be a coset representative for $i,$ therefore, we have
	\begin{equation}\widehat{G}=\widehat{C}=\{\lambda_{\Qu{i}} \mid\Qu{i}\in I \}.
	\end{equation}
	Although the notation for these two types of character groups is the same, we can distinguish them based on the objects they act on. In particular,
	$\lambda_{\Qu{i}}(g)=\lambda_{i}(g)=(-1)^{(i,\bar{g})_s}~(g\in G, ~i \in \mathbb{F}^{2n}_2 ),~	\lambda_{\Qu{i}}(c)=\lambda_{i}(c)=(-1)^{(i,c)_s}~(c \in C ),$ and we also denote $\lambda_e(g)=\lambda_{\bar{e}}(\bar{g})$ as the character of $G.$ It can be observed that $\lambda_{i}=\lambda_{j} \Leftrightarrow \Qu{i}=\Qu{j},$ and $\lambda_{e_1}=\lambda_{e_2} \Leftrightarrow \Qu{\overline{e_{1}}}=\Qu{{\overline{e_{2}}}}.$ Since an $n$-dimensional Hilbert space has an orthogonal direct sum decomposition \cite{feng2010quantum}, 
	$\mathbb{C}^{2^n}=\mathop{\oplus}_{\lambda\in\widehat{G}}Q(\lambda),$ where each $Q(\lambda)=\{v\in \mathbb{C}^{2^n}\mid~\forall~ g\in G,~g(v)=\lambda(g)v\}$ is an invariant subspace of $G$ and can be used to construct additive codes. Let $Q(\lambda_{\Qu{i}}) =Q(\Qu{i})=Q(i),$ then $Q(i)=Q(j)\Leftrightarrow \Qu{i}=\Qu{j},$ therefore, the above decomposition can be written as
	\begin{equation}
		\label{eq:HS sum}
		\mathbb{C}^{2^n}=\mathop{\oplus}_{\Qu{i}\in I }Q(\Qu{i}).
	\end{equation}
	Let $Q(I)=\{Q(\Qu{i})\mid \Qu{i}\in I\},$ the action of quantum errors on the invariant subspaces was described by 
	\begin{thm}
		\label{thm:error action}
		The action of  $E$ on $Q(I)$ corresponds to the additive  of I on itself, and $G^{\bot}_s$ corresponds to $\Qu{\mathbf{0}}.$
	\end{thm}
	\begin{proof}

		For any $|v\rangle \in Q(\Qu{i}),$ and $g\in G,$ we have
		\begin{align*}
			g(e|v\rangle)&=(ge)|v\rangle=(-1)^{(\bar{e},\bar{g})_s}(eg)|v\rangle=(-1)^{(\bar{e},\bar{g})_s}e(g|v\rangle)\\
			&=\lambda_{e}(g)e\lambda_{\Qu{i}} (g)|v\rangle=\lambda_{\Qu{\overline{e}}}(g)\lambda_{\Qu{i}}(g)e|v\rangle\\
			&=\lambda_{\Qu{\overline{e}}+\Qu{i}}(g)e|v\rangle,
		\end{align*}
		where $\lambda_{\Qu{\overline{e}}+\Qu{i}}$ denotes the character associated with the coset $\Qu{\overline{e}}+\Qu{i}.$ Therefore, $e|v\rangle \in Q(\Qu{\overline{e}}+\Qu{i}),$ which means that the action $eQ(\Qu{i})=Q(\Qu{j})$ corresponds to $\Qu{\overline{e}}+\Qu{i}=\Qu{j},$ and  vice versa.
	\end{proof}
	\textit{Quotient Space Codes. }\textit{---}So, the problem of selecting invariant subspaces was transformed into the problem of selecting cosets, which naturally leads to the concept of the quotient space codes. To establish a quotient space code, we first define the distance and norm of the quotient space. For any $\Qu{x},~\Qu{y}\in I,$ we define the distance as 
	\begin{equation}
		d({\Qu{x},\Qu{y}})=\inf_{\tilde{x}\in\Qu{x},\tilde{y}\in\Qu{y}}d(\tilde{x},\tilde{y}),
	\end{equation}
	and define the norm as
	\begin{equation}
		\|\Qu{x}\|=\mathop {\inf}\limits_{ \tilde{x}\in \Qu{x}}\|\tilde{x}\|.\label{eq:min-normal}
	\end{equation}
	When the distance satisfies the translation invariance property, i.e., $d(x+z,y+z)=d(x,y),$ the above definitions satisfy the axioms of a metric and a norm, respectively \cite{rudin1991functional}. To ensure that $\|x\|=w_Q(x)$ becomes a norm, I define $|\alpha |=1$ if $\alpha \neq 0$ ($\alpha \in \mathbb{F}_q$). This allows the quantum weight to satisfy the homogeneity axiom $\|\alpha x\| =|\alpha|\|x\|,$ which is consistent with the definition of a norm. we call this the \textit{quotient minimum norm}. Similarly, the same treatment can be applied to the Hamming weight. Using $d(x,y)=\|x-y\|,$ the distance and norm can be mutually converted, and thus $I$ becomes a metric space and a normed space.
	
	To extend the conclusions of this paper to quantum codes over $q$-ary fields\cite{ashikhmin2001nonbinary,ketkar2006nonbinary}, here is  a general definition of a quotient space code.
	
	Let $V=\mathbb{F}_q^{2n},$  $H$ be a subspace of $V$ and $W=V/H$ be the normed (or metric) quotient space of $V$ with respect to $H.$ If $dim W=n-k,$ then any non-empty subset $\Omega$ of $W$ is called a quotient space code ($ \mathbf{QSC} $) (to avoid confusion with \emph{quantum spherical codes}, it can also be called a \emph{classical QSC}), actually, it is a special coset code \cite{Cosetcode1,Cosetcode2}). Each element (coset) in $\Omega$ is called a codeword. If $|\Omega|=L,$ and $d$ is the minimum distance between cosets in $\Omega,$ then $\Omega$ is called an $(n-k,L,d)_q$ QSC, where $n-k$ represents the code length and $L$ represents the number of codewords. Let $l=\log_qL$ be the number of information bits in the code $\Omega,$ $\Omega$ can also be denoted as  $[n-k,l,d]_q.$	
	
	Similar to classical error-correcting codes \cite{macwilliams1978theory}, a QSC with a distance of $d$ can detect quotient errors up to $d-1,$ and can correct quotient errors up to $\lfloor\frac{d-1}{2} \rfloor.$ 
	
	The distance between $ \Omega $ as QSC and as classical code is related as follows
	\begin{prp}\label{Prp:1}
		Let $\Omega[n-t,k,d]$ be a QSC  ,and $\Omega_{*}$ is the set of elements in the cosets in $ \Omega ,$ if $\delta = d(\Omega_{*})$, then $d \geq \delta$.
	\end{prp}
	
	\begin{proof}
		Let $ d=d(\Omega)=\inf d(\Qu{i},\Qu{j}) =d(i,j), $ where $i,j \in \Omega$. Since $i,j \in \Omega_{*}$, so $d = d(i,j) \geq \delta=d(\Omega_{*})$.
	\end{proof}

	\textit{New  Codes.} \textit{---}The preliminary work is complete, and I now present the new code.
	\begin{thm}
		\label{thm:QSQC} Let $ C $ be a symplectic self-orthogonal subspace in $ \mathbb{F}_2^{2n} ,$ where $ dim C = n-k $ with $ 0\leq k\leq n ,$~let $ I = \mathbb{F}_2^{2n}/{C_s^{\bot}} ,$ and  $ \Omega $ be a QSC in $ I $ with parameters $ (n-k, L, d) ,$ where $ d\leq d_m $ and $ \Omega $ belongs to the same coset in the quotient space $ \mathbb{F}_2^{2n}/C(d-1)_s^{\bot} .$Then, there exists a quantum code $ Q(\Omega) ((n, 2^k L, d)) ,$ where
		\begin{equation}
			Q(\Omega) = \bigoplus_{\Qu{i}\in \Omega} Q(\Qu{i}).
		\end{equation}
	\end{thm}
	\begin{proof}
		
		Because $C \subseteq C_s^{\bot},$ we can lift $C$ to $E_n$ and obtain a stabilizer  $G$ such that $\overline{G} = C.$
		
		First, let us prove that $dim(Q(\Omega)) = 2^k L.$ From equation \eqref{eq:HS sum}, we know that $\mathbb{C}^{2^n} = \bigoplus_{\Qu{i}\in I} Q(\Qu{i}).$ Since the action of $I$ on $I$ is transitive, according to Theorem \ref{thm:error action}, each $Q(\Qu{i})$ is isomorphic and has the same dimension. Hence, $dim(Q(\Qu{i})) = \frac{2^n}{|I|} = \frac{2^n}{2^{n-k}} = 2^k,$ implying $dim(Q(\Omega)) = dim(Q(\Qu{i}))|\Omega| = 2^k L.$
		
		Next, I prove that the minimum distance of $Q(\Omega)$ is $d.$ Let basis vectors $|v_i\rangle, |v_j\rangle \in Q(\Omega),$ where $|v_i\rangle \in Q(\Qu{i})$ and $|v_j\rangle \in Q(\Qu{j}).$ Knill and Laflamme \cite{knill1997theory} and Bennett et al \cite{bennett1996mixed} have shown that a quantum code has minimum distance d  iff
		\begin{equation}
			\langle v_i |e|v_j\rangle = f(e)\delta_{ij},
		\end{equation}
		where  $e \in E_n(d-1)$, $f(e)$ depends only on $e.$ I prove this condition separately for the cases where $\delta_{ij} = 0$ and $\delta_{ij} = 1.$
		
		When $\delta_{ij} = 0,$ the necessary(or distinguishability) condition is given by
		\begin{equation}
			\langle v_i |e|v_j\rangle = 0.
		\end{equation}
		If $\Qu{\bar{e}} = \Qu{\mathbf{0}},$ then $e \in C$ due to $d \leq d_m.$ Hence, $\exists~ g \in G$ such that $e = \gamma^{*}g$ (where $|\gamma| = 1$), and we have $\langle v_i|e|v_j \rangle = \gamma \langle v_i |g|v_j\rangle = \gamma \lambda_{\Qu{j}}(\bar{g}) \langle v_i |v_j\rangle = 0.$ If $\Qu{\bar{e}} \neq \Qu{\mathbf{0}},$ then $0 < \|\Qu{\bar{e}}\| \leq w_Q(\Qu{\bar{e}}) < d.$ Since $\Qu{i}$ and $\Qu{j}$ are codewords in $\Omega$ with distance $d,$ we have
		\begin{equation}\label{eq:necessary}
			\Qu{i} \neq \Qu{\bar{e}} + \Qu{j}.
		\end{equation}
		
		Based on the property that eigenvectors with different eigenvalues are orthogonal, it follows that $\langle v_i |e|v_j\rangle = 0.$ So, the necessary condition is satisfied.
		
		When $\delta_{ij} = 1,$ the measurement condition is 
		\begin{equation}
			\langle v_i |e|v_i\rangle = \langle v_j |e|v_j\rangle = f(e).
		\end{equation}
		If $\Qu{\bar{e}} \neq \Qu{\mathbf{0}},$ then $f(e) = 0.$ If $\Qu{\bar{e}} = \Qu{\mathbf{0}},$ then $e \in C$ due to $d \leq d_m.$ Hence, $\exists~ g \in G$ such that $e = \gamma^{*}g.$ To satisfy $\langle v_i| e|v_i\rangle = \gamma\lambda_{\Qu{i}}(\bar{e}) = \langle v_j| e|v_j\rangle = \gamma\lambda_{\Qu{j}}(\bar{e}) = f(e),$ we need $\lambda_{\Qu{i}}(\bar{e}) = (-1)^{(\Qu{i},\bar{e})_s} = \lambda_{\Qu{j}}(\bar{e}) = (-1)^{(\Qu{j},\bar{e})_s}.$ This implies
		\begin{equation}\label{eq:measuring single e}
			\Qu{i} - \Qu{j} \in (\Qu{\bar{e}})^{\bot}_s,
		\end{equation}
		and since this relation must hold for every $\Qu{\bar{e}} \in C(d-1),$ we have
		\begin{equation}\label{eq:measuring Ed}
			\Qu{i} - \Qu{j} \in C(d-1)^{\bot}_s.
		\end{equation}
		Therefore, when $\Omega$ belongs to the same coset of $\frac{\mathbb{F}^{2n}_2}{C(d-1)^{\bot}_s},$ the measurement condition is satisfied.
	\end{proof}
	
	Define $ Q(\Omega) $ as the \emph{Quotient Space Quantum Code ($\mathbf{{QSQC}}$)}.
	
	From the proof, it is clear that formulas \eqref{eq:necessary} and \eqref{eq:measuring Ed} give the necessary and sufficient conditions for choosing invariant subspaces.
	
	\textit{Encoding circuits.}\textit{---} We can get the basis states of $Q(\Omega)$ by first designing $Q(\Qu{\mathbf{0}})$ using the circuit of additive codes and then obtaining $Q(\Qu{\bar{e}})$ by the operator gate $e$ ($\Qu{\bar{e}} \in \Omega$) action on $Q(\Qu{\mathbf{0}})$.
	
	\textit{Decoding Method.}\textit{---}
	By replacing the classical code part with QSC in  decoding methods  for  some non-additive codes   such as $ ((9,12,3)) $ \cite{yu2008nonadditive} and CWS code \cite{cross2008codeword} , it is believed that the decoding of QSQC can be realised.
	
	\textit{Unique Advantages.}\textit{---} ~If each state in $Q(\Qu{i})$ is considered as the same codeword $\Qu{i},$ then $Q(\Omega)$ can be viewed as a classical code and it's immune to $ G_s^{\bot} .$
	
	\textit{Special Cases .}\textit{---}Next, I discuss  some special cases of Theorem \ref{thm:QSQC}. It is easy to see that when $C(d-1)=\{\mathbf{0}\},$ $Q(\Omega)$ becomes a nondegenerate code. When $L=1,$ we obtain additive codes $((n,2^k,d)).$ When $k=0,$ $C=C_s^{\bot}$ and $(C_s^{\bot})\setminus C=\emptyset,$ therefore, we can remove the restriction $d\leq d_m$ in Theorem \ref{thm:QSQC}, resulting in CWS codes $((n,L,d)).$ Since $C(d-1)\subset C\subset C_s^{\bot}\subset C(d-1)_s^{\bot},$ let $s=\text{dim}(C(d-1))$ and $d_s=\min\{w_Q(c)\mid c\in(C(d-1)_s^{\bot})\setminus C(d-1)\}.$ It is known that the additive code containing $Q(\Omega) $ is  $[[n,n-s,d_s]].$ Let $B(d)=\{ x|w_H(x)\leq d\},$ due to $\overline{E}(d-1)\subset B(2d-2),$ if the distance in equation \eqref{eq:min-normal} is defined as $\|x\|=w_H(x),$ then the distance in Theorem \ref{thm:QSQC} should be adjusted to $2d-1.$

	\textit{New Bounds.}\textit{---} New codes have inspired new bounds. The requirement in Theorem \ref{thm:QSQC} that $\Omega$ belongs to $\mathbb{F}^{2n}_2/C(d-1)^{\bot}_s$ is to satisfy the measurement condition \eqref{eq:measuring Ed}. Due to the translational invariance property of distance, it suffices to find equivalent  QSC within $C(d-1)^{\bot}_s.$ Thus, in contrast to the classical method, in the sphere-packing method the full space to be filled is not $ V $ but $ C(d-1)^{ \bot }_s ,$  the volume of the codeword $ \Qu{i}\in C(d-1)^{ \bot }_s $ after the action of errors is not $ B(\Qu{i},t)=\{\Qu{j}\mid d(\Qu{j}, \Qu{i})\leq t\},$ but $MB(\Qu{i},t)=B(\Qu{i},t)\bigcap C(d-1)^{\bot}_s, $ and the quotient sphere centered at the origin is $\Qu{E}(t)=\{\Qu{\bar{e}}\mid \|\Qu{\bar{e}}\|\leq t\},$ which can be measured as $ME(t)=\Qu{E}(t)\bigcap C(d-1)^{\bot}_s.$
	By the definition of $B(\Qu{i},t),$ it is easy to see that $B(\Qu{i},t)=\Qu{i}+\ \Qu{E}(t).$ Let $\Qu{\bar{e}} \in \Qu{E}(t),$ since $C(d-1)^{\bot}_s$ is a linear space, we have $\Qu{i}+\Qu{\bar{e}}\in C(d-1)^{\bot}_s\Leftrightarrow \Qu{\bar{e}} \in C(d-1)^{\bot}_s\Leftrightarrow \Qu{\bar{e}} \in ME(t).$ Thus, $|C(d-1)^{\bot}_s|=2^{n-k-s}$ and $|MB(\Qu{i},t)|=|ME(t)|,$   where $ |A| $ is the number of different cosets in $\Qu{A}.$  Analogous to deriving bounds of classical code, I obtain two types of quantum bounds--measurement bounds.

	\begin{thm}[Hamming  type]$Q(\Omega)((n,2^kL,d))$ satisfies $2^{k}L|ME(\lfloor \frac{d-1}{2}\rfloor )|\leq 2^{n-s} . $ \label{HB}
	\end{thm}

	\begin{thm}[G-V type]
		If $2^kL|ME(d-1)| < 2^{n-s},$ then there exists a $Q(\Omega)((n,2^{k}(L+1),\geq d)).$ \label{GVB}
	\end{thm}

	When $Q(\Omega)$ is nondegenerate for $d,$ i.e., $C(d-1)^{\bot}_s=V,$ the Hamming bound in Theorem \ref{HB} reduces to the General Hamming bound \cite{EAMC96}. If we can prove that $ \sum_{i=0}^{ \lfloor\frac{d-1}{2}\rfloor}3^{i} { n \choose i}\leq 2^{s}|ME(\lfloor \frac{d-1}{2}\rfloor )|,$ then the Hamming bound for nondegenerate codes can be applied to degenerate nonadditive codes. When $s=n-k,$ we can only construct additive codes.
	
	\textit{Representatives.} \textit{---}To  proof the quantum Singleton bound    \cite{knill1997theory,CC97}, we need to choose appropriate representatives for the cosets.
	
	For each subspace $H,$ there exists a subspace $U$ such that $V=H\oplus U.$ For any vector $v=h+u\in V,$ I define the projection mapping $\mathcal{P}$ as follows
	\begin{equation}
		\mathcal{P}: V \rightarrow U, \quad \mathcal{P}(v) = \mathcal{P}(h+u) = u.
	\end{equation}
	It is clear that $\ker(\mathcal{P})=H,$ and thus $W=V/H\cong U . $ Therefore, each coset $\Qu{x}$ in $W$ has a unique representative $\mathcal{P}x,$ which belongs to $U.$ Let $\vec{T}=(h_1,\dots,h_{n+k},u_1,\dots,u_{n-k})$ be a basis for $V,$ where $h$ is a basis for $H$ and $u$ is a unit basis for $U.$ The matrix  of $\mathcal{P}$ concerning this basis is given by
	\begin{equation}
		P = \mathcal{P}\vec{T} = \begin{pmatrix}
			\mathbf{0}_{n+k} & \mathbf{0} \\
			\mathbf{0} & \mathbf{1}_{n-k}
		\end{pmatrix}.
	\end{equation}
	Let $x=\vec{T}\tilde{x},$ then $Px=(\mathbf{0},\tilde{x}_{n+k+1},\dots,\tilde{x}_{2n})^T.$ Thus, the representatives only need to be represented by $n-k$ bits. To extract the last $n-k$ bits of $Px,$ we can still represent $P$ as $\begin{pmatrix} \mathbf{0} & \mathbf{1}_{n-k} \end{pmatrix},$ then  $Px=(\tilde{x}_{n+k+1},\dots,\tilde{x}_{2n})^T.$ Based on this, I define the \textit{quotient projection norm} as	$	\|\Qu{x}\|_p = w_H(Px)\label{eq:p-norm}.$ It's obvious that $\|\Qu{x}\|_p$ satisfies the norm axioms, and $\|\Qu{x}\|_p \geq w_H(\Qu{x}) = \inf_{\tilde{x}\in \Qu{x}} w_H(\tilde{x}).$ If we want to construct QSQC using this norm, adjustments are needed in Theorem \ref{thm:QSQC}. I now present the Singleton bound.	
	\begin{thm}[Singleton]
		If $Q(\Omega) [[n,k+l,d]]$ satisfies $n\equiv k \mod 2$ and $l\equiv 0 \mod 2,$ then $n\geq k+l+2d-2.$
	\end{thm}
	\begin{proof}
		Since there exists a $\mathbb{F}_2$-linear isomorphism $\psi: \mathbb{F}_2^{2n}\rightarrow \mathbb{F}_4^n$ that preserves inner products \cite{CRSS98}, and $d_H(\psi x,\psi y)=d_Q(x,y),$ we know that there exists $P=\mathcal{P}\circ \psi: \frac{\mathbb{F}_2^{2n}}{C_s^\bot} \rightarrow \mathbb{F}_4^{\frac{n-k}{2}}$ such that $d_H(P\Qu{x},P\Qu{y})\geq d_Q(\Qu{x},\Qu{y}).$ Since $\Omega$ has parameters $[n-k,l,d],$ $P\Omega$ can be regarded as a classical code $[\frac{n-k}{2},\frac{l}{2},t]_4.$ Let $\Qu{i}, \Qu{j}\in \Omega,$ and let $t=d_H(P\Qu{i},P\Qu{j}).$ Since $d_H(P\Qu{i},P\Qu{j})\geq d_Q(\Qu{i},\Qu{j})\geq d, $we have $t\geq d.$ According to the Singleton bound for classical codes, $\frac{n-k}{2}\geq \frac{l}{2}+t-1.$ Hence, $n-k\geq l+2d-2.$
	\end{proof}

	\textit{Examples.}\textit{---}To enhance the  specificity of the theory in this paper, here are two examples. First, I present four constructions of $((8,8,3))$ codes.  We select the self-orthogonal code $C_8,$ also known as the Grassl's code tables  \cite{Grassl:codetables}
	$$C_8= \begin{pmatrix}
		10001011|00101101\\
		01001110|00111010\\
		00101101|01001110\\
		00010111|01011001\\
		00000000|11111111		
	\end{pmatrix}.$$
	By adding the base vectors
	$$\begin{pmatrix}
		\alpha_1\\
		\alpha_2\\
		\alpha_3\\	
	\end{pmatrix}=\begin{pmatrix}
		11101000|00000000\\
		01110100|00000000\\
		11010010|00000000
	\end{pmatrix},$$
	we obtain three additional self-orthogonal codes
	$$
	C_{81}= \begin{pmatrix}
		C_8\\
		\alpha_1		
	\end{pmatrix},\quad
	C_{82}=\begin{pmatrix}
		C_{81}\\
		\alpha_2
	\end{pmatrix},\quad
	C_{83}=\begin{pmatrix}
		C_{82}\\
		\alpha_3\\
	\end{pmatrix}.
	$$
	All four codes have minimum distance $d_m=3$ and are nondegenerate for $d=3.$ Next,  we found the following QSCs
	\begin{align*}
		\Omega_{8}(5,1,3)~ =&\{\Qu{\mathbf{0}}\},\\
		\Omega_{81}(6,2,3)=&\{\Qu{\mathbf{0}},\Qu{01100000|10010000}\},\\
		\Omega_{82}(7,4,3)=&\{
		\Qu{\mathbf{0}},
		\Qu{10100000|11000000},\\
		\Qu{01100000|10010000},~
		&\Qu{11000000|01010000}\},\\
		\Omega_{83}(8,8,3)=&\{\Qu{\mathbf{0}},
		\Qu{10100000|11000000},\\
		\Qu{10010000|10100000},~&
		\Qu{00110000|01100000},\\
		\Qu{01100000|10010000},~&\Qu{11000000|01010000},\\
		\Qu{11110000|00110000},~&
		\Qu{01010000|11110000}
		\}.
	\end{align*}
	By applying theorem \ref{thm:QSQC}, we obtain the following codes: $Q(\Omega_{8})((8,2^{3}\cdot1,3)),$ $Q(\Omega_{81})((8,2^{2}\cdot2,3)),$ $Q(\Omega_{82})((8,2^{1}\cdot4,3)),$ and $Q(\Omega_{83})((8,2^{0}\cdot8,3)).$
	
	The second one, I consider a degenerate CWS code $((12,2,5)).$ The self-dual code is given by
	$$C_{12}=\begin{pmatrix}
		100001000000|010100011000\\
		010001000000|010111001100\\
		001001000000|010010000110\\
		000101000000|011101100100\\
		000011000000|001100110110\\
		000000100010|001100111010\\
		000000010010|010111110110\\
		000000001010|000101010100\\
		000000000110|010100101000\\
		000000000001|000000000000\\
		000000000000|111111000000\\
		000000000000|000000111110\\
	\end{pmatrix}, $$
	with $C_{12}(4)=(000000000001|000000000000).$ We found 
	$$\Omega_{12}(12,2,5)=\{ \Qu{\mathbf{0}}, \Qu{000000111110|000000000000}\},$$
	it is known that $\Omega_{12}  \subset C_{12}(4)^{\bot}_s,$~so we obtain the CWS $Q(\Omega_{12})((12,2,5)).$

	\textit{Relation to USt codes.}\textit{---} USt codes considered the construction of codes with multiple invariant subspaces of the stabilizer $ G  $ as the same as this paper, but our approaches are fundamentally different.
	
	The QSQC framework first  compute $  C(d-1)^{\bot}_s $,~for given  distance d, and then selects  QSC  $\Omega$ of distance $ d $ among $C(d-1)^{\bot }_s\backslash C^{\bot }_s $~(or its translation).
	
	However, the USt method first selects cosets   $\Omega$  and then calculates the distance\begin{equation}
		d=\min\{d(v):v\in (\Omega-\Omega)\backslash (\Omega^{\bot}_s \cap C)\},
	\end{equation}
	where $ \Omega-\Omega =\{a-b\mid a,b\in \Omega  \}$ denotes the set of all differences of vectors (rather than cosets) in  $ \Omega   $ ~\cite{GR2013}. So the distance depends on the selection of cosets and every time a new set of cosets is selected, it is necessary to compute $ (\Omega^{\bot}_s \cap C) $ to recalculate d which makes this method complex. It does not use the distance of the coset but still relies on the distance of the vector elements, which limits the theory's ability.
	It is known in the framework of QSQC that this order can easily lead to a small or even zero distance in the case of degenerate codes i.e two  cosets do not among $C(d-1)^{\bot }_s\backslash C^{\bot }_s $~(or its translation), while in the case of nondegenerate codes, $ C(d-1)^{\bot }_s=V ,$   there is no need to worry about such a situation, it is only necessary to select cosets whose distance is $ d $.  Therefore, the USt construction is the same as QSQC of  nondegenerate codes, while it is difficult to deal with degenerate codes, which is also a problem for some non-additive codes such as CWS codes.
	
	Another difference is that the QSQC method uses the concept of coset error correction more essentially, so that different quotient norms, such as the projection norm, can be defined for coset to achieve various purposes of quantum error correction.
	
	Below, I give  degenerate   codes that  satisfy  a strict inequality in Eq. (10.6) in \cite{GR2013} .	
	First, let's select the self-orthogonal code $C_9$ in Grassl's code tables \cite{Grassl:codetables}
	$$C_9= \begin{pmatrix}
		100010000|000011110\\
		000101000|100001000\\
		010011010|000110000\\
		000101110|010110010\\
		001110100|000101110\\
		000000110|001000100\\
		000000001|000000000		
	\end{pmatrix},$$
	with 	$  d(C_9)=1,d_m=3,$ so it is degenerate  for $ d=2 $. Next,we get 	$C_9(1)=\{000000001|000000000\},
	$	and select a QSC 
	\begin{align*}
		\Omega_9(7,16,2)~ =&\{\Qu{000000000|000000000},\\
		\Qu{110000000|000000000},~&
		\Qu{101000000|000000000},\\
		\Qu{011000000|000000000},~&
		\Qu{100100000|000000000},\\
		\Qu{010100000|000000000},~&
		\Qu{001100000|000000000},\\
		\Qu{111100000|000000000},~&
		\Qu{100010000|000000000},\\
		\Qu{010010000|000000000},~&
		\Qu{001010000|000000000},\\
		\Qu{111010000|000000000},~&
		\Qu{000110000|000000000},\\
		\Qu{110110000|000000000},~&
		\Qu{101110000|000000000},\\ \Qu{011110000|000000000}~~&\}.
	\end{align*}
	
	It is known that $\Omega_{9}  \subset C_{9}(1)^{\bot}_s,$~so we get $Q(\Omega_{9})((9,2^2\cdot 16,2)),$~i.e USt code$ [[9, 6, 2]] $.
	There is a strict inequality 
	$$ d(\Omega_9) = 2 >  d({\Omega_9}_{*}) = 1,$$where   $ d({\Omega_9}_{*}) $ is the minimum distance of $ \Omega_9 $ as a classical union normalizer code.
	
	If we remove an arbitrary coset  in $\Omega_{9}$, we can  obtain a truly degenerate nonadditive code $ ((9,60,2)) $ that satisfies a strict inequality in Eq. (10.6) in  \cite{GR2013}.
	
	Indeed, since $|C(1)^{\bot}_s|=2^{6},|ME(1)|=25,$ and $2\cdot 25< 2^{6}$, there must be a degenerate USt code $((9,2^{2}\cdot 3,\geq 2))$ according to theorem (\ref{GVB}) without selecting specific QSC.
	
	Based on the same method we can obtain  degenerate USt code $[[7, 4, 2]],$ where
	$$ 	C_7=\begin{pmatrix}
		1011100|0000000\\
		0000000|1111110\\
		0100010|0000000\\
		0000001|0000000
	\end{pmatrix},C_7(1)=\{
0000001|0000000\}$$  and	
	$ \Omega_7(4,2,2)=
	\{\Qu{0000000|0000000},
	\Qu{0000000|1100000}\}$ with  $ d(\Omega_7)=d_m=2>d(C_7)=1. $ 
	
	\textit{Conclusion and Discussion.}\textit{---}In conclusion, I  established  quotient space codes $(n-k,L,d)$ to construct quotient space  quantum codes $((n,2^kL,d)),$ which includes additive codes and codeword stabilized codes as special cases and can transmit $ L $ classical codewords. I have used this new framework to derive new bounds that are applicable to both degenerate and nondegenerate codes and to show a simple proof of the Singleton bound.  However, I did not find a simple method to compute the distance of QSC   and it is also not clear whether there are strictly nonadditive codes that are both non-CWS codes. If we consider each invariant subspace as a stone, additive codes use the largest stone for tiling, while CWS uses the smallest stone for tiling, QSQC encompasses tiling schemes with stones of various sizes. As the volume of the stones decreases, the gaps between them become smaller, but the number of stones to be selected increases. Therefore, intuitively, CWS codes can achieve the optimal parameters but have the highest algorithmic complexity. On the other hand,	QSQCs are upgraded to additive codes  and can be  easily extended to asymmetric, non-binary, subsystem, entanglement-assisted  codes, and so on.

	\textit{Acknowledgement}.\textit{---}I would like to thank Shanzhen Chen for his academic guidance, which encouraged me to study QECCs, and Keqin Feng and Hao Chen for the book \cite{feng2010quantum}, whose proofs of additive codes inspired me to write this paper. It is a pleasure to acknowledge the
	feedback from Victor V. Albert and  Markus Grassl.

	\bibliographystyle{apsrev4-2}

\end{document}